\documentclass[11pt,a4paper,reqno]{amsart}%
\usepackage{amsthm,amsmath,amsfonts,amssymb,amsxtra,appendix,bookmark,dsfont,bm,mathrsfs}

\theoremstyle{plain}
\newtheorem{theorem}{Theorem}
\newtheorem{lemma}[theorem]{Lemma}

\newtheorem{remark}[theorem]{Remark}

\theoremstyle{definition}

\newcounter{step} 

\usepackage{etoolbox}
\AtBeginEnvironment{proof}{\setcounter{step}{0}}

\def\bq{\begin{eqnarray}}
\def\eq{\end{eqnarray}}
\def\bqq{\begin{align*}}
\def\eqq{\end{align*}}

\def\nn{\nonumber}
\def\minus {\backslash}

\renewcommand{\Re}{\operatorname{Re}}

\newcommand\1{{\ensuremath {\mathds 1} }}

\def\cF {\mathcal{F}}

\def\cH{\mathcal{H}}

\def\R {\mathbb{R}}

\def\cN {\mathcal{N}}

\def\cE {\mathcal{E}}

\def\gH{\mathcal{H}}

\def\R {\mathbb{R}}

\def\d{{\rm d}}
\def\bT{{\mathbb{T}}}
\def\bZ{{\mathbb{Z}}}

\newcommand{\bH}{\mathbb{H}}

\title[] {Binding energy of homogeneous Bose gas}

\author[P. T. Nam]{Phan Th\`anh Nam}
\address{Department of Mathematics and Statistics, Masaryk University, Kotl\'a\v rsk\'a 2, 611 37 Brno, Czech Republic} 
\email{ptnam@math.muni.cz}

\begin{document}

\begin{abstract}   We compute the energy needed to add or remove one particle from a homogeneous system of $N$ bosons on a torus. We focus on the mean-field limit when $N$ becomes large and the strength of particle interactions is proportional to $N^{-1}$, which allows us to justify Bogoliubov's approximation.  
\end{abstract}

\date{\today}

\maketitle

\setcounter{tocdepth}{1}

\section{Introduction}

In 1947, Bogoliubov \cite{Bogoliubov-47b} proposed an important approximation theory for the low-energy behavior of weakly interacting Bose gases. Bogoliubov's approximation has been justified rigorously in many cases, including the ground state energy of one and two-component Bose gases  \cite{LieSol-01,LieSol-04,Solovej-06}, the Lee-Huang-Yang formula of dilute gases \cite{ErdSchYau-08,GiuSei-09,YauYin-09}, and the excitation spectrum of mean-field Bose gases  \cite{Seiringer-11,GreSei-13,LewNamSerSol-15,DerNap-13,NamSei-15}.

In this note, we shall show that Bogoliubov's approximation can be used to compute the {\em binding energy}, which is the energy needed to add or remove one particle from the system. An interesting feature of our calculation is that the difference between two quantities can be computed up to an accuracy which goes beyond the separate knowledge of each quantity. 

We consider a homogeneous system of $N$ bosons on the unit torus $\bT^d$, $d\ge 1$. The system is governed by the Hamiltonian
$$
H_{\lambda,N} = \sum_{j=1}^N -\Delta_{x_j} +\lambda \sum_{1\le j<k\le N} w(x_j-x_k)
$$
which acts on the bosonic Hilbert space
$$\gH^N=L^2_{\rm sym}((\bT^d)^N).$$
Here the kinetic operator $-\Delta$ is the usual Laplacian (with periodic boundary condition). The interaction potential $w:\mathbb{T}^d \to \mathbb{R}$ is an even, bounded function with nonnegative Fourier transform 
$$\widehat w(p)=\int_{\mathbb{T}^d} w(x)e^{-ip\cdot x} \d x \ge 0, \quad \forall p\in (2\pi \mathbb{Z})^d.$$ 
The parameter $\lambda>0$ stands for the strength of interactions. We will focus on the weak interaction/mean-field regime
$$ \lambda\sim N^{-1}.$$

Let $E(\lambda,N)$ be the ground state energy of $H_{\lambda,N}$. It was proved in \cite{Seiringer-11} that when $N\to \infty$ and $\lambda N \to 1$, 
\bq \label{eq:EN-GSE}
E(\lambda,N)= \frac{\lambda}{2} N(N-1) \widehat w(0) + e_{\rm B} + o(1)
\eq
where 
\bq \label{eq:Bog-GSE}
e_{\rm B} = -  \frac{1}{2} \sum_{p\in (2\pi \mathbb{Z})^d \minus \{0\}} \left( |p|^2 + \widehat w(p) - \sqrt{|p|^4+2 |p|^2 \widehat w(p)}\right).
\eq
Note that the sum in \eqref{eq:Bog-GSE} is finite since the summands is smaller than $|\widehat w(p)|^2$. In fact, $e_{\rm B}$ is the ground state energy of a quadratic Hamiltonian obtained from Bogoliubov's approximation (as we will explain). The energy expansion \eqref{eq:EN-GSE} was proved later in \cite{LewNamSerSol-15} by a different method, which gives also the convergence of the ground state.  See also \cite{GreSei-13,NamSei-15,DerNap-13,RouSpe-16,BocBreCenSch-17} for related results.  

In this note, we are interested in the binding energy $E(\lambda, N-1)-E(\lambda, N)$. From \eqref{eq:EN-GSE}, one gets immediately that 
\bq \label{eq:binding-leading}
E(\lambda, N)-E(\lambda, N-1)= \lambda (N-1) \widehat w(0)  + o(1).
\eq
Our new result is the next order expansion in \eqref{eq:binding-leading}. We have 

\begin{theorem}[Binding energy] \label{thm:main} When $N\to \infty$ and $\lambda N \to 1$, we have 
\bq \label{eq:main-result}
E(\lambda, N)-E(\lambda, N-1)=\lambda (N-1)\widehat w(0) + \frac{1}{N} \left( e_{\rm B} - \sum_{p\ne 0}  \frac{|p|^2 \alpha_p^2}{1-\alpha_p^2} + o(1)\right)
\eq
where
\bq \label{eq:alphap}
\alpha_p= \frac{\widehat w(p)}{p^2+ \widehat w(p) + \sqrt{|p|^4+ 2|p|^2 \widehat w(p)}}.
\eq
\end{theorem}

The proof of Theorem \ref{thm:main} is based on a trial state argument inspired from the study in \cite{Bach-91} on the ionization energy of bosonic atoms. While the analysis in \cite{Bach-91} focuses on the leading order, here we can go to the next order by a careful analysis using Bogoliubov's approximation. 

Note that for  the homogeneous Bose gas, the condensate is described by the constant function and this enables several simplifications. We leave open the question of generalizing Theorem \ref{thm:main} to inhomogeneous systems. See Section \ref{sec:general} for further discussion. 

The rest of the note is organized as follows. In Section \ref{sec:Bog-app} we will quickly recall Bogoliubov's approximation. In Section \ref{sec:GSE} we give a short proof of \eqref{eq:EN-GSE}, using ideas from \cite{Seiringer-11,LewNamSerSol-15} with some simplifications. Then we prove Theorem \ref{thm:main} in Section \ref{sec:main-thm}. Further extension is discussed in Section  \ref{sec:general}. 

\medskip

\noindent\textbf{Acknowledgement.} I would like to thank Mathieu Lewin and Robert Seiringer for helpful discussions. 

\section{Bogoliubov approximation} \label{sec:Bog-app}

It is convenient to turn to the grand canonical setting. Let us introduce the Fock space 
$$
\cF=\mathbb{C} \oplus \gH \oplus \gH^2 \oplus \cdots = \bigoplus_{n=0}^\infty \gH^n.
$$
Let $a_p^*=a^*(u_p)$ and $a_p=a(u_p)$ be the usual creation and annihilation operators on Fock space with $u_p(x)=e^{ipx}$, $p\in (2\pi \mathbb{Z})^d$. They satisfy the canonical commutation relation (CCR)
\bq \label{eq:CCR}
[a_p,a_q]=0=[a_p^*,a_q^*], \quad [a_p,a_q^*]=\delta_{p,q}.
\eq
Henceforth, we denote 
$$[A,B]=AB-BA.$$  
We can extend $H_{\lambda,N}$ to a grand canonical Hamiltonian on Fock space as
\bq \label{eq:grand-Hamiltonian}
\bH_\lambda=\sum_{p}  |p|^2 a_p^* a_p +\frac{\lambda}{2} \sum_{p,q,\ell } \widehat w(\ell) a^*_{p-\ell}a^*_{q+\ell}a_p a_q.
\eq
Here the sums are always taken over the momentum space $(2\pi \mathbb{Z})^d$. 

Bogoliubov's approximation consists of two simplifications. First, substituting all operators $a_0$ and $a_0^{*}$ in \eqref{eq:grand-Hamiltonian} by the scalar number $\sqrt{N}$ \footnote{Strictly speaking, for $a_0^*a_0^*a_0a_0$ we should rewrite it as $(a_0^*a_0)^2-a_0^*a_0$ before doing the substitution} (this so-called {\em c-number substitution} is motivated by the Bose-Einstein condensation $\langle a_0^* a_0 \rangle \sim N$). Second, ignoring all interaction terms which are coupled with constants $\lambda \sqrt{N}$ or $\lambda$ (because $\lambda \sim N^{-1}$). All this leads to the formal expression
\bq \label{eq:Bog-formal}
\bH_\lambda \approx \frac{1}{2}\lambda N(N-1) \widehat w(0) + \bH_{\rm B}
\eq
where $\bH_{\rm B}$ is the quadratic Hamiltonian (so-called Bogoliubov Hamiltonian)
$$
\bH_{\rm B}= \sum_{p\ne 0} \left( \Big( |p|^2+\widehat w(p)  \Big)  a_p^* a_p + \frac{1}{2}\widehat w(p) \Big( a^*_p a^*_{-p}+ a_p a_{-p}\Big)\right).
$$

Note that $\bH_{\rm B}$ can be restricted naturally to the Fock space of excited particles 
$$
\cF_+=\mathbb{C} \oplus \gH_+ \oplus \gH_+^2 \oplus \cdots = \bigoplus_{n=0}^\infty \gH_+^n, \quad \gH_+= \{u_0\}^\bot.
$$

A fundamental fact observed by Bogoliubov \cite{Bogoliubov-47b} is that the quadratic Hamiltonian $\bH_{\rm B}$ is {\em diagonalizable}. More precisely, we can write 
\bq \label{eq:Bog-dia}
 \bH_{\rm B} = e_{\rm B} + \sum_{p\ne 0} e_p c_p^* c_p 
\eq
where $e_{\rm B}$ is given in \eqref{eq:Bog-GSE}, $e_p$ is the {elementary excitation}
$$
e_p= \sqrt{|p|^4+2|p|^2 \widehat w(p)}
$$
and $c_p$ is the {modified annihilation operator}
\bq \label{eq:cp}
c_p= \frac{a_p + \alpha_p a_{-p}^*}{\sqrt{1-\alpha_p^2}}
\eq
(with $\alpha_p$ defined by \eqref{eq:alphap}).
Consequently, $e_{\rm B}$ is the ground state energy of $\bH_{\rm B}$. Moreover, $\bH_{\rm B}$ has a unique ground state
\bq \label{eq:Phi1}
\Phi^{(1)}= U_{\rm B} |0\rangle \in \cF_+
\eq
where $|0\rangle = 1 \oplus 0 \oplus 0 \cdots $ is the vacuum state in Fock space and $U_{\rm B}$ is a unitary operator on $\cF_+$ satisfying 
\bq \label{eq:Bog-trans}
U_{\rm B}^* c_p U_{\rm B}=a_p, \quad U_{\rm B}^* a_p U_{\rm B} = \frac{a_p-\alpha_p a_{-p}^*}{\sqrt{1-\alpha_p^2}} \quad \forall p\ne 0.
\eq

In the next section, we will justify \eqref{eq:Bog-formal} rigorously, and in particular give a proof of \eqref{eq:EN-GSE}.

\section{Ground state properties} \label{sec:GSE}

In our proof of Theorem \ref{thm:main}, we need not only the convergence of energy \eqref{eq:EN-GSE} but also the convergence of ground state. To discuss the latter, as in \cite{LewNamSerSol-15}, we introduce the unitary operator 
$$U_N: \gH^N\to \cF_+^{\le N}=\bigoplus_{n=0}^N \cH_+^n$$ 
which is defined by 
\begin{align} \label{eq:UN}
U_N = \sum_{j=0}^N Q^{\otimes j} \left(\frac{a_0^{N-j}}{\sqrt{(N-j)!}} \right), \quad U_N^*= \bigoplus_{j=0}^N \left(\frac{(a_0^*)^{N-j}}{\sqrt{(N-j)!}} \right).
\end{align}
Here $Q=\1-|u_0\rangle \langle u_0|$. The role of $U_N$ is to implement the c-number substitution. To be precise, we have the identities (see \cite[Proposition 4.2]{LewNamSerSol-15})
\begin{align} \label{eq:action}
U_N a_p^* a_p U_N^*&=  a_p^* a_p,\quad U_N a_p^* a_0  U_N^* =  a_p^* \sqrt{N-\cN_+}, \quad \forall p \ne 0,
\end{align}

In this section, we recall the following results, taken from \cite{Seiringer-11} and  \cite{LewNamSerSol-15}. 

\begin{theorem}[Ground state properties] \label{thm:GSE} When $N\to \infty$ and $\lambda N \to 1$, then \eqref{eq:EN-GSE} holds true. Moreover, if $\Psi_N$ is the ground state of $H_{\lambda,N}$ and $\Phi^{(1)}$ is the ground state of the Bogoliubov Hamiltonian $\bH_{\rm B}$, then $U_N \Psi_N \to \Phi^{(1)}$ in the quadratic form domain of $\bH_{\rm B}$.
\end{theorem}

For the reader's convenience, we will give a simplified proof of Theorem \ref{thm:GSE} below. Some parts of the proof will be needed for the proof of Theorem \ref{thm:main}. 

Note that by subtracting $w$ by a constant, we can always assume that
\bq
\widehat w (0)=0.
\eq

Let us start with a preliminary result. 

\begin{lemma}[Number of excited particles] \label{lem:cN+} Let $\cN_+=\sum_{p\ne 0} a_p^* a_p$.   If $\Psi_N$ is the ground state of $H_{\lambda,N}$, then 
$$ \langle \Psi_N, \cN_+^k \Psi_N \rangle \le C, \quad k=1,2.$$ 
\end{lemma}

\begin{proof} First, the variational principle gives us immediately the upper bound
\bq \label{eq:leading-upper}
 \langle \Psi_N, H_{\lambda,N} \Psi_N \rangle = E_{\lambda,N} \le \langle 1, H_{\lambda,N} 1\rangle =0.
 \eq
On the other hand, since $\widehat w\ge 0$ we have
\begin{align*}
\sum_{1\le j<k \le N} w(x_j-x_k) &= \frac{1}{2}\sum_{j,k=1}^N w(x_j-x_k) -  \frac{N}{2} w(0)  \\
&=\sum_{p\in (2\pi \mathbb{Z})^d} \frac{\widehat w(p)}{2} \left| \sum_{j=1}^N e^{ipx_j}\right|^2 -  \frac{N}{2} w(0) \ge -  \frac{N}{2} w(0) .
\end{align*}
Thus 
\bq \label{eq:leading-lower}
H_{\lambda,N} \ge \sum_{p\in (2\pi\bZ)^d} |p|^2 a_p^* a_p - \frac{\lambda N}{2}w(0) \ge (2\pi)^{2d} \cN_+ -C.
\eq
From \eqref{eq:leading-upper} and \eqref{eq:leading-lower}, we obtain
\bq \label{eq:bound-cN+}
|E(\lambda,N)| \le C \quad \text{and} \quad\langle \Psi_N, \cN_+ \Psi_N \rangle \le C.
\eq

Now we are going to control $\cN_+^2$. Since $\Psi_N$ is a ground state of $H_{\lambda,N}$, it solves the Schr\"odinger equation 
$$ H_{\lambda,N}\Psi_N= E(\lambda,N)\Psi_N.$$
Consequently, we get the identity 
\begin{align} \label{eq:identity-cN+-HN}
\left\langle \Psi_N, \cN_+ \Big( H_{\lambda,N} - E(\lambda,N)  \Big) \cN_+ \Psi_N \right\rangle = - \Big\langle \Psi_N,   \cN_+ [H_{\lambda,N}, \cN_+]   \Psi_N \Big\rangle.
\end{align}
The left side of \eqref{eq:identity-cN+-HN} can be estimated using \eqref{eq:leading-lower} and \eqref{eq:bound-cN+} as
\bq \label{eq:identity-cN+-HN-left}
 \left\langle \Psi_N, \cN_+ \Big( H_{\lambda,N} - E(\lambda,N)  \Big) \cN_+ \Psi_N \right\rangle \ge \left\langle \Psi_N, \Big( (2\pi)^{2}   \cN_+^3 -C\cN_+^2 \Big) \Psi_N \right\rangle.
\eq
For the right side of \eqref{eq:identity-cN+-HN}, we use \eqref{eq:grand-Hamiltonian} and the CCR \eqref{eq:CCR} to write 
\begin{align} \label{eq:identity-cN+-HN-right-0}
 \cN_+[H_{\lambda,N}, \cN_+] & =\frac{\lambda}{2}  \sum_{\ell \ne 0} \sum_{p,q } \widehat w(\ell) \cN_+[a^*_{p-\ell}a^*_{q+\ell}a_p a_q, \cN_+] \nn\\
&=\frac{\lambda}{2} \sum_{\ell \ne 0} \widehat w(\ell) \cN_+ \Big( 2 a^*_0 a^*_0 a_{\ell} a_{-\ell} -2 a^*_{-\ell} a^*_{\ell} a_0 a_0   \Big) \nn\\
& + \frac{\lambda}{2} \sum_{\ell \ne 0 \ne p \ne \ell} \widehat w(\ell) \cN_+ \Big(a^*_{p-\ell}a^*_{0}a_p a_{-\ell}  - a^*_{p-\ell}a^*_{\ell}a_p a_0  \Big) \nn\\
& + \frac{\lambda}{2} \sum_{\ell \ne 0 \ne q \ne -\ell} \widehat w(\ell) \cN_+ \Big(a^*_{0}a^*_{q+\ell}a_{\ell} a_{q}  - a^*_{-\ell}a^*_{q+\ell}a_0 a_q  \Big).
\end{align}
Now we take the expectation against $\Psi_N$ and estimate. For the first term on the right side of \eqref{eq:identity-cN+-HN-right-0}, by the Cauchy--Schwarz inequality, we get
\begin{align*}
& \left| \left\langle \Psi_N, \sum_{\ell \ne 0} \widehat w(\ell) \cN_+  a^*_0 a^*_0 a_{\ell} a_{-\ell} \Psi_N \right\rangle \right| \le \sum_{\ell \ne 0}   \left\|  a_0 a_0 \cN_+ \Psi_N \right\| | \widehat w(\ell)| \left\| a_{\ell} a_{-\ell} \Psi_N\right\| \\
&\le \left\|  a_0 a_0 \cN_+ \Psi_N \right\|   \left( \sum_{\ell \ne 0} |\widehat w(\ell)|^2 \right)^{1/2} \left( \sum_{\ell \ne 0}  \left\| a_{\ell} a_{-\ell} \Psi_N\right\| ^2 \right)^{1/2} \le C N \langle \Psi_N, \cN_+^2 \Psi_N\rangle.
\end{align*}
Here we have used $a_0^*a_0 \le N$ on $\gH^N$ and $\sum |\widehat w(\ell)|^2=\|w\|_{L^2}^2<\infty$. Similarly, for the second term, we have 
\begin{align*}
&\left| \left\langle \Psi_N, \sum_{\ell \ne 0} \widehat w(\ell) \cN_+  a^*_{-\ell} a^*_\ell a_0 a_0 \Psi_N \right\rangle \right| = \left| \left\langle \Psi_N, \sum_{\ell \ne 0} \widehat w(\ell)   a^*_{-\ell} a^*_\ell  (\cN_++2)  a_0 a_0\Psi_N \right\rangle \right| \\
&\le \left\| (\cN_+ +2) a_0 a_0   \Psi_N \right\|   \left( \sum_{\ell \ne 0} |\widehat w(\ell)|^2 \right)^{1/2} \left( \sum_{\ell \ne 0}  \left\| a_{\ell} a_{-\ell} \Psi_N\right\| ^2 \right)^{1/2} \\
&\le C N \langle \Psi_N, (\cN_+ +2)^2 \Psi_N\rangle.
\end{align*}
For the third term, we can bound
\begin{align*}
&\left| \left\langle \Psi_N, \sum_{\ell \ne 0 \ne p \ne \ell}  \widehat w(\ell) \cN_+ a^*_{p-\ell}a^*_{0}a_p a_{-\ell}  \Psi_N \right\rangle \right| \\ 
&\le \sum_{\ell \ne 0 \ne p \ne \ell} |\widehat w(\ell)| \| a_{0} a_{p-\ell} \cN_+ \Psi_N\| \|a_p a_{-\ell}  \Psi_N \|  \\
&\le \left( \sum_{\ell \ne 0 \ne p \ne \ell}  |\widehat w(\ell)|^2 \left\| a_{0} a_{p-\ell} \cN_+ \Psi_N \right\|^2  \right)^{1/2} \left(   \sum_{\ell \ne 0 \ne p \ne \ell} \left\| a_p a_{-\ell}  \Psi_N \right\|^2 \right)^{1/2} \\
& \le C N^{1/2}\langle \Psi_N, \cN_+^3 \rangle^{1/2} \langle \Psi_N, \cN_+^2 \Psi_N\rangle^{1/2}
\end{align*}
and proceed similarly for other terms. Thus in summary, from \eqref{eq:identity-cN+-HN-right-0} we get
\begin{align} \label{eq:identity-cN+-HN-right}
\left| \Big\langle \Psi_N, \cN_+[H_{\lambda,N},\cN_+] \Psi_N \Big\rangle \right| &\le C \langle \Psi_N, (\cN_+ +2)^2 \Psi_N\rangle \nn\\
&\quad + C N^{-1/2}\langle \Psi_N, \cN_+^3 \rangle^{1/2} \langle \Psi_N, \cN_+^2 \Psi_N\rangle^{1/2}.
\end{align}
Here we have used $\lambda \sim N^{-1}$. Inserting \eqref{eq:identity-cN+-HN-left} and \eqref{eq:identity-cN+-HN-right} into \eqref{eq:identity-cN+-HN}, we obtain
\begin{align*}
\left\langle \Psi_N, \Big( (2\pi)^{2}   \cN_+^3 - C\cN_+^2 \Big) \Psi_N \right\rangle &\le C \langle \Psi_N, (\cN_+ +2)^2 \Psi_N\rangle \nn\\
&\quad + C N^{-1/2}\langle \Psi_N, \cN_+^3 \rangle^{1/2} \langle \Psi_N, \cN_+^2 \Psi_N\rangle^{1/2}.
\end{align*}
By the Cauchy-Schwarz inequality and the a-priori estimate $\langle \Psi_N, \cN_+ \Psi_N\rangle \le C$ in \eqref{eq:bound-cN+}, we can conclude that $\langle \Psi_N, \cN_+^k \Psi_N \rangle \le C$ for $k=1,2,3.$ 
\end{proof}

\begin{remark}
By adapting the above proof we can show that $\langle \Psi_N, \cN_+^k \Psi_N \rangle\le C_k$ for all $k\in \mathbb{N}$. However, having control $\cN_+^2$ is sufficient for our purpose.
\end{remark}

\begin{proof}[Proof of Theorem \ref{thm:GSE}] Let $\Psi_N$ be a ground state for $H_N$ and let $\Phi_N=U_N \Psi_N$. We will prove that 
\begin{align} \label{eq:Bog-appro-rigorous}
\langle \Psi_N, H_{\lambda,N} \Psi_N \rangle = \langle \Phi_N, \bH_{\rm B}\Phi_N \rangle + o(1). 
\end{align}
As explained, we can assume $\widehat w (0)=0$. Thus 
$$
\langle \Psi_N, H_{\lambda,N}\Psi_N\rangle = \sum_{p \ne 0} \langle \Psi_N, |p|^2 a_p^* a_p \Psi_N\rangle  + \frac{\lambda}{2} \sum_{\ell\ne 0}\sum_{p,q} \widehat w(\ell) \langle \Psi_N, a^*_{p-\ell}a^*_{q+\ell}a_p a_q \Psi_N\rangle.
$$

Note that in the interaction terms, if $\ell\ne0$, then the number of zeros among 4 numbers $p,q,p+\ell,q-\ell$ is either 0,1 or 2.  

When all numbers $p,q,p+\ell,q-\ell$ are non-zero, we have 
\begin{align*} 
&\frac{\lambda}{2}\left| \left\langle \Psi_N, \sum_{p \ne 0, p\ne \ell, q\ne 0, q\ne -\ell} \widehat w(\ell) a^*_{p-\ell}a^*_{q+\ell}a_p a_q \right\rangle \right| \nn\\
&= \lambda \left| \left\langle \Psi_N, \sum_{1\le j \le N} (Q\otimes Q w Q\otimes Q)_{jk} \Psi_N\right\rangle \right| \nn\\
& \le C \lambda \|w\|_{L^\infty} \langle \Psi_N, \cN_+^2 \Psi_N \rangle  \le CN^{-1}. 
\end{align*}
Here we have used $\sum_{j=1}^N Q_j=\cN_+$ on $\gH^N$ and Lemma \ref{lem:cN+}.

When there is exactly one zero among $p-\ell, q+\ell, p, q$, for example $q=0$, we have 
\begin{align} \label{eq:R3}
&\lambda \left| \sum_{p \ne 0 \ne \ell \ne p} \widehat w(\ell) \langle \Psi_N, a^*_{p-\ell}a^*_{\ell}a_p a_0 \Psi_N \rangle \right| \nn\\
& \le \lambda \left( \sum_{p \ne 0\ne \ell \ne p}  \| a_{p-\ell} a_{\ell} \Psi_N \|^2  \right)^2 \left(  \sum_{p \ne 0 \ne \ell \ne p} |\widehat w(\ell)|^2 \| a_p a_0 \Psi_N\|^2 \right)^2 \nn \\
& \le C \lambda \langle \Psi_N, \cN_+ \Psi_N \rangle^{1/2} \langle \Psi_N, \cN \cN_+ \Psi_N\rangle^{1/2} \le CN^{-1/2}.
\end{align}

Thus only the interactions terms with two zeros among $p-\ell, q+\ell, p, q$ have important energy contribution. In other words, we have proved that 
\begin{align} \label{eq:lw-0}
\langle \Psi_N, H_{\lambda,N} \Psi_N \rangle &= \sum_{p\ne 0} \left\langle \Psi_N,  |p|^2 a_p^* a_p \Psi_N \right\rangle + O(N^{-1/2})\\
& + \frac{\lambda}{2} \sum_{p\ne 0} \widehat w(p) \left\langle \Psi_N, (2a^*_{p} a_p  a_0^*a_0+ a^*_{p} a^*_{-p} a_0 a_0 + a_0^* a_0^* a_{p} a_{-p}) \Psi_N \right\rangle.\nn
\end{align}

Now using $\Phi_N=U_N \Psi_N$ and \eqref{eq:action}, we can translate \eqref{eq:lw-0} into
\begin{align}  \label{eq:lw-1}
& \langle \Psi_N, H_{\lambda,N} \Psi_N \rangle =  \langle \Phi_N, \bH_{\rm B}\Phi_N\rangle + \sum_{p\ne 0} \widehat w(p) \left\langle \Phi_N,  a_p^* a_p \Big(\lambda N- \lambda \cN_+ -1 \Big)   \Phi_N \right\rangle \nn\\
&+  \Re \sum_{p \ne 0} \widehat w(p)  \left\langle \Phi_N, a_p^* a_{-p}^* \Big( \lambda\sqrt{(N-\cN_+)(N-\cN_+-1)} -1 \Big) \Phi_N \right\rangle + O(N^{-1/2}). 
\end{align}

Next, from \eqref{eq:action} and Lemma \ref{lem:cN+} we have
\bq \label{eq:lem:cN+}
\langle \Phi_N, \cN_+^k \Phi_N\rangle = \langle \Psi_N, \cN_+^k \Psi_N\rangle\le C,\quad k=1,2.
\eq
Therefore, using $0\le \widehat w(p) \le C$ and $\lambda N \to 1$, we can estimate 
\begin{align} \label{eq:lw-1a}
&\left|  \sum_{p\ne 0} \widehat w(p) \left\langle \Phi_N,  a_p^* a_p \Big(\lambda N- \lambda \cN_+ -1 \Big)   \Phi_N \right\rangle \right| \nn\\
& \le  \sum_{p\ne 0} \widehat w(p) \left\langle \Phi_N,  a_p^* a_p \Big( |\lambda N-1| + CN^{-1} \cN_+ \Big) \Phi_N \right\rangle \nn\\
&\le  C |\lambda N-1| \langle \Phi_N, \cN_+ \Phi_N\rangle + C N^{-1} \langle \Phi_N, \cN_+^2 \Phi_N\rangle = o(1).
\end{align}
Moreover, using the operator inequality 
\begin{align*}
 \Big|\lambda \sqrt{(N-\cN_+)(N-\cN_+-1)}-1\Big|^2  &\le \left| \lambda^2 (N-\cN_+)(N-\cN_+-1) -1 \right|^2 \\
&\le C \Big( |\lambda^2 N^2-1|^2 + (\cN_++1)^2 \Big)
\end{align*}
on $\cF_+^{\le N}$ and the Cauchy--Schwarz inequality, we can estimate 
\begin{align} \label{eq:lw-1b}
& \left| \sum_{p\ne 0} \widehat w(p) \left  \langle \Phi_N, a_p^* a_{-p}^*\Big( \lambda\sqrt{(N-\cN_+)(N-\cN_+-1)}-1\Big) \Phi_N \right\rangle\right| \nn\\
& \le  \left( \sum_{p\ne 0} \| a_p a_{-p} \Phi_N\|^2 \right)^{1/2} \left( \sum_{p\ne 0} |\widehat w(p)|^2  \right)^{1/2} \times \nn\\
&\qquad \qquad \qquad \times \Big\| \Big( \lambda \sqrt{(N-\cN_+)(N-\cN_+-1)}-1 \Big) \Phi_N \Big\|\nn \\
&\le C \left\langle \Phi_N, \cN_+^2 \Phi_N \right\rangle^{1/2} \left\langle \Phi_N, \Big( |\lambda^2 N^2-1|^2 + (\cN_++1)^2 \Big)  \Phi_N \right\rangle = o(1).
\end{align}

Inserting \eqref{eq:lw-1a} and \eqref{eq:lw-1b} into \eqref{eq:lw-1}, we obtain \eqref{eq:Bog-appro-rigorous}. 

From \eqref{eq:Bog-appro-rigorous}, we deduce immediately the lower bound
\bq \label{eq:lw-ene}
E(\lambda,N)=\langle \Psi_N, H_N \Psi_N \rangle = \langle \Phi_N, \bH_{\rm B}\Phi_N \rangle + o(1) \ge e_{\rm B} + o(1).
\eq

Now we turn to the upper bound. Recall that the ground state $\Phi^{(1)}$ of $\bH_{\rm B}$ satisfies \eqref{eq:Phi1}-\eqref{eq:Bog-trans}. It is straightforward to check that if $\1^{\le N}$ is the projection onto $\cF_+^{\le N}$, then
$$
\| \1^{\le N} \Phi^{(1)}\| \to 1, \quad \langle \1^{\le N} \Phi^{(1)} ,  \bH_{\rm B}\1^{\le N} \Phi^{(1)} \rangle \to e_{\rm H}
$$
and
$$
\langle \1^{\le N} \Phi^{(1)}, \cN_+^k \1^{\le N} \Phi^{(1)} \rangle \le \langle  \Phi^{(1)}, \cN_+^k  \Phi^{(1)} \rangle \le C, \quad k=1,2.
$$
Therefore, by repeating the proof of  \eqref{eq:Bog-appro-rigorous} with $\Psi_N$ replaced by $
\widetilde \Psi_N = U_N^* \widetilde \Phi_N$ we obtain the upper bound
\bq \label{eq:up-ene}
E(\lambda,N) \le \frac{\langle \widetilde \Psi_N, H_N \widetilde \Psi_N \rangle}{\|\widetilde \Psi_N\|^2} = \frac{\langle \widetilde \Phi_N, \bH_{\rm B}\widetilde \Phi_N \rangle + o(1)}{\|\widetilde \Phi_N\|^2} = e_{\rm B}+ o(1).
\eq

Combining \eqref{eq:lw-ene} and \eqref{eq:up-ene}, we conclude that \eqref{eq:EN-GSE} holds true and  
\bq \label{eq:cv-PhiN}
\langle \Phi_N, \bH_{\rm B}\Phi_N \rangle \to e_{\rm B}.
\eq
Since $\bH_{\rm B}$ has a spectral gap to the ground state energy $e_{\rm B}$, which can be seen from \eqref{eq:Bog-dia}, we deduce from \eqref{eq:cv-PhiN} that  $\Phi_N \to \Phi^{(1)}$ in the quadratic form domain of $\bH_{\rm B}$.
\end{proof}


\section{Binding energy} \label{sec:main-thm}

Now we are ready to provide the

\begin{proof}[Proof of Theorem \ref{thm:main}] Subtracting $w$ by a constant if necessary, we can assume that $\widehat w(0)=0$. Recall that $H_{\lambda,N}$ can be extended to Fock space as
\bq \label{eq:bHl}
\bH_\lambda= \sum_{p} |p|^2 a_p^* a_p + \frac{\lambda}{2} \sum_{\ell\ne 0}\sum_{p,q} \widehat w(\ell) a^*_{p-\ell}a^*_{q+\ell}a_p a_q.
\eq

\noindent {\bf Lower bound.} Let $\Psi_N$ be the ground state of $H_{\lambda,N}$ and denote $\Phi_N=U_N \Psi_N$. Using the Schr\"odinger equation 
$$
\bH_\lambda \Psi_N=  E(\lambda,N)\Psi_N
$$
and the variational inequality
$$
E(\lambda,N-1)\le \frac{\langle a_0 \Psi_N, \bH_\lambda a_0 \Psi_N \rangle}{\|a_0 \Psi_N\|^2}
$$
we obtain the lower bound 
\begin{align} \label{eq:binding-lw-0}
E(\lambda,N)- E(\lambda,N-1)&\ge \frac{\langle a_0^*a_0 \Psi_N, \bH_\lambda \Psi_N \rangle}{\|a_0 \Psi_N\|^2} - \frac{\langle a_0 \Psi_N, \bH_\lambda a_0 \Psi_N \rangle}{\|a_0 \Psi_N\|^2} \nn\\
& = \frac{\langle \Psi_N, [\bH_\lambda,a_0^*] a_0 \Psi_N \rangle}{\|a_0 \Psi_N\|^2} .
\end{align}

Let us estimate the right side of \eqref{eq:binding-lw-0}. By Lemma \ref{lem:cN+}, 
\bq \label{eq:binding-lw-1}
\|a_0 \Psi_N\|^2 =  N - \langle \Phi_N, \cN_+ \Phi_N\rangle = N + O(1).
\eq
Moreover, from \eqref{eq:bHl} and the CCR \eqref{eq:CCR},
\begin{align*}
[\bH_\lambda,a_0^*] a_0 &= \frac{\lambda}{2} \sum_{\ell\ne 0} \sum_{p,q} \widehat w(\ell) a^*_{p-\ell}a^*_{q+\ell}[a_p a_q,a_0^*] a_0 \\
&= \frac{\lambda}{2} \sum_{\ell\ne 0} \sum_{p,q} \widehat w(\ell) a^*_{p-\ell}a^*_{q+\ell}( a_p \delta_{q,0}+ a_q \delta_{p,0}) a_0 \\
&= \lambda  \sum_{p \ne 0} \widehat w(p) \Big( a^*_p a_p a^*_0 a_0 +  a^*_{p} a^*_{-p} a_0 a_0 \Big) +  \lambda \sum_{p\ne 0 \ne \ell \ne p} \widehat w(\ell)  a^*_{p-\ell} a^*_ \ell a_p a_0 .
\end{align*}
Next, we take the expectation again $\Psi_N$. From the estimates \eqref{eq:R3},  \eqref{eq:lw-1a} and \eqref{eq:lw-1b} in the proof of Theorem \ref{thm:GSE}, we have
$$\lambda  \left|  \sum_{p\ne 0 \ne \ell \ne p} \widehat w(\ell) \left\langle \Psi_N, a^*_{p-\ell} a^*_ \ell a_p a_0  \Psi_N \right\rangle \right| \le  CN^{-1/2}
$$
and
\begin{align*}
& \lambda \sum_{p \ne 0} \widehat w(p) \left\langle \Psi_N, \Big( a_p^* a_p a_0^* a_0 + a^*_{p} a^*_{-p} a_0 a_0 \Big) \Psi_N \right\rangle \\
&=  \sum_{p \ne 0} \widehat w(p) \left\langle \Phi_N, \Big(  a^*_{p} a_{p}  +  a^*_{p} a^*_{-p}  \Big) \Phi_N \right\rangle + o(1).
\end{align*}
Therefore,
\begin{align} \label{eq:binding-lw-2a}
\langle \Psi_N, [\bH_\lambda,a_0^*] a_0 \Psi_N \rangle  = \sum_{p \ne 0} \widehat w(p) \left\langle \Phi_N, \Big( a^*_{p} a_{p} +   a^*_{p} a^*_{-p}    \Big) \Phi_N \right\rangle + o(1) .
\end{align}
From \eqref{eq:binding-lw-0}, we know that $\langle \Psi_N, [\bH_\lambda,a_0^*] a_0 \Psi_N \rangle$ is a real number. Therefore, we can rewrite \eqref{eq:binding-lw-2a} as
\begin{align}  \label{eq:binding-lw-2b}
\langle \Psi_N, [\bH_\lambda,a_0^*] a_0 \Psi_N \rangle &= \sum_{p \ne 0} \widehat w(p) \left\langle \Phi_N, \Big(   a^*_{p} a_{p} + \frac{1}{2}a^*_{p} a^*_{-p} +  \frac{1}{2}a_{p} a_{-p} \Big) \Phi_N \right\rangle + o(1)\nn\\
&= \langle  \Phi_N, \bH_{\rm B} \Phi_N \rangle - \sum_{p\ne 0} \langle \Phi_N, |p|^2 a_p^*a_p \Phi_N\rangle + o(1).
\end{align}

Recall that 
$$\langle  \Phi_N, \bH_{\rm B} \Phi_N \rangle \to e_{\rm B}$$
(see \eqref{eq:cv-PhiN}) and $\Phi_N$ converges to the ground state $\Phi^{(1)}$ of $\bH_{\rm B}$ in the quadratic form domain of $\bH_{\rm B}$. On the other hand, by a calculation similar to \eqref{eq:Bog-dia} with $w$ replaced by $2w$, we have
\begin{align*}
&\frac{1}{2}\sum_{p\ne 0} |p|^2 a_p^* a_p + \sum_{p \ne 0} \widehat w(p) \Big(   a^*_{p} a_{p} + \frac{1}{2}a^*_{p} a^*_{-p} +  \frac{1}{2}a_{p} a_{-p} \Big) \\
&\ge -  \frac{1}{4} \sum_{p\in (2\pi \mathbb{Z})^d \minus \{0\}} \left( |p|^2 + 2\widehat w(p) - \sqrt{|p|^4+4 |p|^2 \widehat w(p)}\right) \ge -C.
\end{align*}
Therefore,
$$
\bH_{\rm B} \ge \frac{1}{2}\sum_{p\ne 0} |p|^2 a_p^* a_p  - C.
$$
Thus $\sum_{p\ne 0} p^2 a_p^* a_p$ is  bounded relatively to $\bH_{\rm B}$, and hence
$$
\sum_{p\ne 0} \langle \Phi_N, |p|^2 a_p^*a_p \Phi_N\rangle \to \sum_{p\ne 0} \langle \Phi^{(1)}, |p|^2 a_p^*a_p \Phi^{(1)}\rangle.
$$
Finally, using \eqref{eq:Phi1} and \eqref{eq:Bog-trans}, it is straightforward to compute 
\begin{align*}
& \sum_{p\ne 0} \big\langle \Phi^{(1)}, |p|^2 a_p^*a_p \Phi^{(1)}\big\rangle = \sum_{p\ne 0} |p|^2 \big\langle 0\big| U_{\rm B}^* a_p^*a_p U_{\rm B}\big| 0 \big\rangle \\
&= \sum_{p\ne 0}  |p|^2 \left\langle 0\left| \left( \frac{a_p-\alpha_p a_{-p}^*}{\sqrt{1-\alpha_p^2}} \right)^* \left( \frac{a_p-\alpha_p a_{-p}^*}{\sqrt{1-\alpha_p^2}} \right) \right| 0 \right\rangle = \sum_{p\ne 0}  \frac{|p|^2 \alpha_p^2}{1-\alpha_p^2} .
\end{align*}
Here the sum is finite because $|p|^2\alpha_p^2 \le |\widehat w(p)|^2$. In summary, from \eqref{eq:binding-lw-2b}  we deduce that 
\begin{align} \label{eq:binding-lw-2}
\langle \Psi_N, [\bH_\lambda,a_0^*] a_0 \Psi_N \rangle = e_{\rm B} - \sum_{p\ne 0}  \frac{|p|^2 \alpha_p^2}{1-\alpha_p^2}  + o(1).
\end{align}

Inserting \eqref{eq:binding-lw-1} and \eqref{eq:binding-lw-2} into \eqref{eq:binding-lw-0}, we arrive that the lower bound
\begin{align} \label{eq:lw}
E(\lambda,N)- E(\lambda,N-1) \ge \frac{1}{N} \left( e_{\rm B} - \sum_{p\ne 0}  \frac{|p|^2 \alpha_p^2}{1-\alpha_p^2} + o(1)\right).
\end{align}

\noindent{\bf Upper bound.} Let $\Psi_{N-1}$ be the the ground state of $H_{\lambda,N-1}$. By the variational principle, we have
\begin{align} \label{eq:up-0}
E(\lambda,N)- E(\lambda,N-1)&\le \frac{\Big\langle a_0^* \Psi_{N-1}, \bH_\lambda a_0^* \Psi_{N-1}\Big\rangle}{\|a_0^* \Psi_{N-1}\|^2} - \frac{\Big\langle a_0 a_0^* \Psi_{N-1}, \bH_\lambda \Psi_{N-1}\Big\rangle}{\|a_0^* \Psi_{N-1}\|^2}   \nn\\
&= \frac{\langle \Psi_{N-1}, a_0 [\bH_\lambda,a_0^*] \Psi_{N-1} \rangle}{\|a_0^* \Psi_{N-1}\|^2} .
\end{align}
Using Lemma \ref{lem:cN+} with $\Psi_N$ replaced by $\Psi_{N-1}$, we get
\bq \label{eq:up-1}
\|a_0^* \Psi_{N-1}\|^2 = 1 + \| a_0 \Psi_{N-1}\|^2 = N - \langle \Psi_{N-1}, \cN_+ \Psi_{N-1} \rangle = N + O(1).
\eq
Next, we use the identity
\begin{align} \label{eq:up-2a}
a_0 [\bH_\lambda,a_0^*] - [\bH_\lambda,a_0^*] a_0 &=  [a_0,[\bH_\lambda,a_0^*]]  \nn\\
&= \frac{\lambda}{2} \sum_{\ell\ne 0} \sum_{p,q} w(\ell) [a_0,a^*_{p-\ell}a^*_{q+\ell}] \cdot [a_p a_q,a_0^*] \nn\\
&= \frac{\lambda}{2} \sum_{\ell\ne 0} \sum_{p,q} w(\ell) ( a^*_{p-\ell} \delta_{q+\ell,0} + a^*_{q+\ell} \delta_{p-\ell,0})( a_p \delta_{q,0}+ a_q \delta_{p,0}) \nn\\
&= \frac{\lambda}{2}  \sum_{p \ne 0} \widehat w(p) \Big( a_p^* a_p + a^*_{p} a^*_{-p} \Big)
\end{align}
and take the expectation against $\Psi_{N-1}$. Using \eqref{eq:binding-lw-2a} and \eqref{eq:binding-lw-2} with $\Psi_N$ replaced by $\Psi_{N-1}$ (and $\Phi_N$ replaced by $\Phi_{N-1}=U_{N-1}\Psi_{N-1}$), we have
\begin{align*}
\langle \Psi_{N-1}, [\bH_\lambda,a_0^*] a_0 \Psi_{N-1}\rangle &= \sum_{p \ne 0} \widehat w(p) \left\langle \Phi_{N-1}, \Big( a_p^* a_p + a^*_{p} a^*_{-p} \Big) \Phi_{N-1}\right\rangle + o(1) \\
&=  e_{\rm B} - \sum_{p\ne 0}  \frac{|p|^2 \alpha_p^2}{1-\alpha_p^2} + o(1).
\end{align*}
Therefore, from \eqref{eq:up-2a} it follows that
\begin{align} \label{eq:up-2}
\langle \Psi_{N-1}, a_0 [\bH_\lambda,a_0^*] \Psi_{N-1}\rangle &= \langle \Psi_{N-1}, [\bH_\lambda,a_0^*] a_0 \Psi_{N-1}\rangle \nn\\
&\quad + \frac{\lambda}{2} \sum_{p \ne 0} \widehat w(p) \left\langle \Phi_{N-1}, \Big( a_p^* a_p + a^*_{p} a^*_{-p} \Big) \Phi_{N-1}\right\rangle \nn\\
&= \Big( 1 +\frac{\lambda}{2} \Big) \Big( e_{\rm B} - \sum_{p\ne 0}  \frac{|p|^2 \alpha_p^2}{1-\alpha_p^2} + o(1) \Big) \nn\\
&= e_{\rm B} - \sum_{p\ne 0}  \frac{|p|^2 \alpha_p^2}{1-\alpha_p^2} + o(1).
\end{align}
Here we have used $\lambda \sim N^{-1}$ in the latter estimate. 

Inserting \eqref{eq:up-1} and \eqref{eq:up-2} into \eqref{eq:up-0} we obtain the upper bound
\begin{align} \label{eq:up}
E(\lambda,N)- E(\lambda,N-1)\le \frac{1}{N}\left( e_{\rm B} - \sum_{p\ne 0}  \frac{|p|^2 \alpha_p^2}{1-\alpha_p^2}+ o(1)\right).
\end{align} 
The desired result \eqref{eq:main-result} follows from \eqref{eq:lw} and \eqref{eq:up}. 
\end{proof}

\section{Extension} \label{sec:general}

It is natural to ask for an extension of Theorem \ref{thm:main} to the more general case. To be precise, let us consider a system of $N$ bosons in $\Omega\subset \R^d$ described by the Hamiltonian
$$
H_{\lambda,N} =\sum_{i=1}^N T_i +\lambda \sum_{i<j}^N w(x_i-x_j)
$$
on $L^2_{\rm sym}(\Omega^N)$. Here $T>0$ is a self-adjoint operator on $L^2(\Omega)$, $w:\R^d \to \R$ is an even function  bounded relatively to $T$ and $\lambda\sim N^{-1}$. A typical example is $\Omega=\R^3$, $T=-\Delta+V(x)$ is $w(x)=|x|^{-1}$.

In principle, the condensate $u_0$ is the minimizer of the Hartree functional
$$
\cE_{\rm H} (u)= \langle u, Tu \rangle + \frac{\lambda (N-1)}{2}\iint |u(x)|^2 w(x-y) |u(y)|^2 \d x \d y .
$$
Consequently, it solves the Euler-Lagrange equation
\bq \label{eq:EL}
(T+|u_0|^2*w-\mu_{\rm H})u_0=0
\eq
with some constant $\mu_{\rm H}\in \mathbb{R}$ (the chemical potential).  

Following the strategy in Section \ref{sec:main-thm}, we extend $H_{\lambda,N}$ to Fock space as
$$
\bH_{\lambda} = \sum_{m,n} \langle u_m, T u_n \rangle a_m^* a_n + \frac{1}{2}\langle u_m \otimes u_n, w u_p \otimes u_q\rangle a_m^* a_n^* a_p a_q
$$
where $\{u_n\}_{n=0}^\infty$ is an orthonormal basis for $L^2(\Omega)$ and $a_n=a(u_n)$. Similarly to \eqref{eq:binding-lw-0}, we have
\begin{align} \label{eq:EN-EN-1>abc}
E(\lambda,N)- E(\lambda,N-1)\ge \frac{\langle \Psi_N, [\bH_\lambda,a_0^*] a_0 \Psi_N \rangle}{\|a_0 \Psi_N\|^2}. 
\end{align}
Note that $\|a_0 \Psi_N\|^2=N+o(N)$ because $u_0$ is the condensate. However, if we calculate $\langle \Psi_N, [\bH_\lambda,a_0^*] a_0 \Psi_N \rangle$, then beside the terms similar to the homogeneous case, we have the new terms
$$
A= \sum_{n\ne 0} \Big( \langle u_n , T u_0\rangle \langle \Psi_N, a_n^* a_0 \Psi_N\rangle + 2\lambda \langle u_n,(|u_0|^2*w) u_0 \rangle \langle \Psi_N, a_n^* a_0^* a_0 a_0 \Psi_N\rangle\Big).
$$

It is not clear to us whenever $A$ is small, namely of order $o(1)$. Heuristically, if we use Bogoliubov's c-number substitution and equation \eqref{eq:EL}, then we arrive at the approximation  
\bq \label{eq:B}
A \approx  - \sqrt{N} \sum_{n\ne 0} \langle u_n , T u_0\rangle \langle \Phi_N, a_n^* \Phi_N\rangle.
\eq
where $\Phi_N=U_N \Psi_N$ with $U_N$ defined in \eqref{eq:UN}. On the right side of \eqref{eq:B}, if we replace $\Phi_N$ by the ground state $\Phi^{(1)}$ of the Bogoliubov Hamilonian, then we obtain $0$ since $\Phi^{(1)}$ is a quasi-free state.  However, to prove rigorously that $A$ is of order $o(1)$, we need a very strong quantitative estimate for the convergence $\Phi_N\to \Phi^{(1)}$, which seems out of reach for existing methods.

\end{document}